\documentclass[11pt]{article}

\usepackage{graphicx}
\usepackage{amsmath,amssymb,amsthm}

\usepackage{thm-restate}

\usepackage{tikz}
\usetikzlibrary{arrows,decorations.pathmorphing,backgrounds,positioning,fit,petri}

\usepackage[bookmarksnumbered,colorlinks,linkcolor=blue,citecolor=blue,urlcolor=blue,
  pdftitle={Entanglement can increase asymptotic rates of zero-error classical communication over classical channels},
  pdfauthor={Debbie Leung, Laura Mancinska, William Matthews, Maris Ozols, Aidan Roy},
]{hyperref}

\newcommand{\spg}[2]{sp(#1,#2)}
\newcommand{\ipb}[2]{\langle #1, #2 \rangle}

\newcommand{\vre}[1]{\phi(#1)}
\newcommand{\set}[1]{\{#1\}}

\newcommand{\defemph}[1]{\textbf{#1}}
\newcommand{\KG}[2]{KG_{#1,#2}}
\newcommand{\lc}[1]{\tilde{#1}}

\def\onev{\mathbf{1}}
\def\RR{\mathbb{R}}
\def\CC{\mathbb{C}}
\def\FF{\mathbb{F}}
\def\ZZ{\mathbb{Z}}

\def\sf{\sigma}
\def\al{\alpha}

\def\ox{\otimes}
\def\mc{\mathcal}
\def\tp{^{\mathsf{T}}} 
\DeclareRobustCommand\idop{\leavevmode\hbox{\small1\normalsize\kern-.33em1}}
\def\1{\idop}
\def\w{w}
\def\de{\mathbf{E}}

\DeclareMathOperator{\Tr}{Tr}

\DeclareMathOperator{\rank}{rank}

\theoremstyle{definition}
\newtheorem{theorem}{Theorem}
\newtheorem{lemma}[theorem]{Lemma}
\newtheorem{definition}[theorem]{Definition}
\newtheorem{remark}[theorem]{Remark}

\begin{document}

\title{Entanglement can increase asymptotic rates of zero-error classical communication over classical channels}
\author{Debbie Leung, Laura Mancinska, William Matthews, Maris Ozols, Aidan Roy\\Institute of Quantum Computing, University of Waterloo, \\200 University Ave. West, Waterloo, ON N2L 3G1, Canada}

\maketitle

\begin{abstract}
	It is known that the number of different classical messages which can be communicated with a \emph{single use} of a classical channel with zero probability of decoding error can sometimes be increased by using entanglement shared between sender and receiver. It has been an open question to determine whether entanglement can ever increase the zero-error communication rates achievable in the limit of many channel uses. In this paper we show, by explicit examples, that entanglement can indeed increase asymptotic zero-error capacity, even to the extent that it is equal to the \emph{normal} capacity of the channel. Interestingly, our examples are based on the exceptional simple root systems $E_7$ and $E_8$.
\end{abstract}

\section{Introduction}

A \emph{classical channel} $\mc{N}$ which is discrete and memoryless is fully described by its conditional probability distribution $\mc{N}(y|x)$ of producing output $y$ for a given input $x$. The channel obtained by allowing one use of a channel $\mc{N}_1$ and one use of $\mc{N}_2$ is written as $\mc{N}_1\ox\mc{N}_2$, reflecting the fact that its conditional probability matrix is the tensor (or Kronecker) product of those of the two constituent channels. Similarly, $\mc{N}^{\ox n}$ denotes $n$ uses of $\mc{N}$.

\begin{definition}
	Let $M_0(\mc{N})$ denote the maximum number of different messages which can be sent with a single use of $\mc{N}$ with zero probability of a decoding error. The \defemph{zero-error capacity} of $\mc{N}$ is
	\begin{equation}
		C_0(\mc{N}) := \lim_{n\to\infty} \frac{1}{n} \log M_0(\mc{N}^{\ox n}).
	\end{equation}
\end{definition}

Two input symbols $x_1, x_2$ of a channel $\mc{N}$ are said to be \emph{confusable} if $\mc{N}(y|x_1) > 0$ and $\mc{N}(y|x_2) > 0$ for some output symbol $y$. The \emph{confusability graph} of a channel $\mc{N}$ is a graph $G(\mc{N})$, whose vertices correspond to different input symbols of $\mc{N}$ and two vertices are joined by an edge if the corresponding symbols are confusable. The confusability graph of $\mc{N}_1\ox\mc{N}_2$ is determined by those of $\mc{N}_1$ and $\mc{N}_2$ as follows: $G(\mc{N}_1\ox\mc{N}_2) = G(\mc{N}_1)\boxtimes G(\mc{N}_2)$, where ``$\boxtimes$'' denotes the \emph{strong graph product}.

\begin{definition}[Strong graph product]
	In general, the \defemph{strong product} of graphs $G_1, \ldots, G_n$ is a graph $G_1\boxtimes \cdots \boxtimes G_n$, whose vertices are the $n$-tuples $V(G_1) \times \cdots \times V(G_n)$ and distinct vertices $(a_1, \ldots, a_n)$ and $(b_1, \ldots, b_n)$ are joined by an edge if they are entry-wise confusable, i.e., for each $j \in \set{1,\ldots,n}$ either $a_j b_j \in E(G_j)$ or $a_j = b_j$. Likewise, we define the \defemph{strong power} of graph $G$ by $G^{\boxtimes 1} := G$ and $G^{\boxtimes n} := G\boxtimes G^{\boxtimes (n-1)}$.
\end{definition}

An \emph{independent set} of a graph is a subset of its vertices with no edges between them. The \emph{independence number} $\alpha(G)$ of a graph $G$ is the maximum size of an independent set of $G$. As Shannon observed~\cite{Shannon}, $M_0$ and $C_0$ depend only on the confusability graph of the channel: $M_0(\mc{N}) = \alpha(G(\mc{N}))$ and $C_0(\mc{N}) = \log \Theta(G(\mc{N}))$
where
\begin{equation}
	\Theta(G) := \lim_{n\to\infty} \sqrt[n]{\alpha(G^{\boxtimes n})}
\end{equation}
is known as the \emph{Shannon capacity} of the graph $G$. Clearly, $\Theta(G) \geq \alpha(G)$, since $G^{\boxtimes n}$ has an independent set of size $\alpha(G)^n$. However, in general $\Theta(G)$ can be larger than $\alpha(G)$. The simplest example is the $5$-cycle $C_5$ for which $\alpha(C_5) = 2$ but $\Theta(C_5) = \sqrt{5}$.

Computing the independence number of a graph is NP-hard, but conceptually simple. However, no algorithm is known to determine $\Theta(G)$ in general, although there is a celebrated upper bound due to Lov\'asz \cite{Lovasz}. He defined an efficiently computable quantity $\vartheta(G)$ called the \emph{Lov\'asz number} of $G$ which satisfies $\vartheta(G) \geq \alpha(G)$ and $\vartheta(G_1\boxtimes G_2) = \vartheta(G_1)\vartheta(G_2)$. Because of these properties we also have $\vartheta(G) \geq \Theta(G)$.

\begin{definition}
	Let $M_0^{\mathrm{E}}(\mc{N})$ denote the number of different messages which can be sent with a single use of $\mc{N}$ with zero probability of a decoding error, when both parties share an arbitrary finite-dimensional entangled state on which each can perform arbitrary local measurements. The \defemph{entanglement-assisted zero-error capacity} of $\mc{N}$ is
	\begin{equation}
		C_0^{\mathrm{E}}(\mc{N}) := \lim_{n\to\infty} \frac{1}{n} \log M_0^{\mathrm{E}}(\mc{N}^{\ox n}).
	\end{equation}
\end{definition}

As in the unassisted case, the quantities $M_0^{\mathrm{E}}$ and $C_0^{\mathrm{E}}$ also depend \emph{only} on the confusability graph of the channel~\cite{CLMW}. For this reason, we will talk about the assisted and unassisted zero-error capacities of \emph{graphs} as well as of channels.

In~\cite{CLMW} it was shown that graphs $G$ exist with $M_0^{\mathrm{E}}(G) > M_0(G)$.  Shortly afterwards it was shown~\cite{Beigi,DSW} that the Lov\'asz bound also applies to the entanglement-assisted quantities, so $M_0^{\mathrm{E}}(G) \leq \vartheta(G)$ and hence $C_0^{\mathrm{E}}(G) \leq \log \vartheta(G)$.

Whether graphs with $C_0^{\mathrm{E}}(G) > C_0(G)$ exist, that is, whether entanglement can ever offer an advantage in terms of the \emph{rates} achievable in the \emph{large block length limit} was left as an open question. Clearly, the Lov\'asz bound cannot be used to prove such a separation. Fortunately, there is another bound due to Haemers which is sometimes better than the Lov\'asz bound \cite{Haemers,Haemers2}.

\begin{theorem}[Haemers~\cite{Haemers,Haemers2}]\label{Haemer}
	For $u,v \in V(G)$ let $M_{uv}$ be a matrix with entries in any field $K$. We say that $M$ \defemph{fits} $G$ if $M_{uu} \neq 0$ and $M_{uv} = 0$ whenever there is no edge between $u$ and $v$. Then $\Theta(G) \leq R(G) := \min \set{ \rank(M) : M \text{ fits } G }$. In particular, $C_0(G) \leq \log R(G)$.
\end{theorem}
\begin{proof}
	Let $S$ be a maximal independent set in $G$. If $M$ fits $G$, then $M_{uv} = 0$ for all $u \neq v \in S$ while the diagonal entries are non-zero. Hence, $M$ has full rank on a subspace of dimension $|S|$ and thus $\rank(M) \geq |S| = \alpha(G)$. As this is true for any $M$ that fits $G$, we get $R(G) \geq \alpha(G)$.

	Next, note that if $M_1$ fits $G_1$ and $M_2$ fits $G_2$ then $M_1 \ox M_2$ fits $G_1 \boxtimes G_2$, and $\rank(A \ox B) = \rank(A) \rank(B)$. Hence, $R(G_1 \boxtimes G_2) \leq R(G_1)R(G_2)$, which implies the desired result.
\end{proof}

The next section shows how Haemers bound applies to a particular graph to determine its unassisted zero-error capacity, and then provides an explicit entanglement-assisted protocol that achieves a higher rate. This shows that entanglement assistance can indeed increase the asymptotic zero-error rate, thus giving an affirmative answer to the previously open question.

The entanglement-assisted protocol is based on the fact that the graph in question can be constructed from the \emph{root system} \cite{roots} $E_7$, so in Section 3 we investigate constructions based on other irreducible root systems. Most notably we show that a construction based on $E_8$ provides another example with a larger gap in the capacities. In Section 4 we discuss open problems.

\section{The zero-error capacities of the symplectic graph $\spg{6}{\FF_2}$} \label{sect:Main}


\begin{definition}[Symplectic space]
	A \defemph{non-degenerate symplectic space} $(V,S)$ is a vector space $V$ (over a field $K$) equipped with a \defemph{non-degenerate symplectic form}, i.e., a bilinear map $S: V \times V \to K$ which is
\begin{itemize}
	\item skew-symmetric: $S(u,v) = -S(v,u)$ for all $u,v \in V$, and
	\item non-degenerate: if $S(u,v) = 0$ for all $v \in V$, then $u = 0$.
\end{itemize}
	If $K$ has characteristic 2, we also require that $S(u,u) = 0$ for all $u \in V$ (for other fields this is implied by the anti-symmetry property). On a $2m$-dimensional space, the \defemph{canonical symplectic form} is
	\begin{equation}
		\sf(u,v) := u^T \left( \begin{array}{cc} 0 & \1_m \\ -\1_m & 0 \end{array} \right)v.
	\end{equation}
	where $\1_m$ is the $m \times m$ identity matrix.
Any non-degenerate symplectic space with finite dimensional vector space $V$ is isomorphic to the canonical symplectic space $(V,\sf)$.
\end{definition}

\begin{definition}[Symplectic graph]
	Let $K$ be a finite field and let $m$ be a natural number. The vertices of the \defemph{symplectic graph} $\spg{2m}{K}$ are the points of the projective space $\mathbb{P} K^{2m}$ and there is an edge between $u, v \in \mathbb{P} K^{2m}$ if $\sf(u, v) = 0$. In the case where $K = \FF_2$, the points of the projective space are simply the $2^{2m} - 1$ non-zero elements of $\FF_2^{2m}$.
\end{definition}
\begin{remark}
	The symplectic graph $\spg{2m}{\FF_2}$ is isomorphic to the graph whose vertices are all the $m$-fold tensor products of Pauli matrices except for the identity, i.e., $\set{\1,X,Y,Z}^{\ox m} \setminus \set{1^{\ox m}}$, and which has edges between commuting matrices.
\end{remark}
The next two subsections prove that for channels with the confusability graph $\spg{6}{\FF_2}$, the entanglement-assisted zero-error capacity is larger than the unassisted capacity. More precisely,
\begin{theorem}\label{mainthm}
	$C_0\big(\spg{6}{\FF_2}\big) = \log 7$ while $C_0^{\mathrm{E}}\big(\spg{6}{\FF_2}\big) = \log 9$.
\end{theorem}

\subsection{Capacity in the unassisted case}

The fact that $C_0\big(\spg{6}{\FF_2}\big) = \log 7$ is a special case of a result in~\cite{Peeters} which we prove explicitly here.

\begin{theorem}[Peeters~\cite{Peeters}]\label{sc}
	$C_0\big(\spg{2m}{\FF_2}\big) = \log (2m + 1)$.
\end{theorem}
\begin{proof}
	For the upper bound we construct a matrix over $\FF_2$ which fits $\spg{2m}{\FF_2}$ and which has rank $2m+1$ and use Haemers' bound (see Theorem \ref{Haemer}). Let
	\begin{equation}
		U_m := \{ v \in \FF_2^{2m+1} : \ipb{v}{v} = 0 \}
	\end{equation}
be the $2m$-dimensional subspace of $\FF_2^{2m+1}$ that consists of vectors which have an even number of entries equal to one. The restriction of the standard inner product $\ipb{\cdot}{\cdot}$ on $\FF_2^{2m+1}$ to the subspace $U_m$ is a non-degenerate symplectic form, so there is an isomorphism $T: (\FF_2^{2m}, \sf) \to (U_m, \ipb{\cdot}{\cdot})$ such that
	\begin{equation}
		\forall u,v \in \FF_2^{2m}: \sf(u,v) = \ipb{T(u)}{T(v)}.
	\end{equation}
	Let $\onev$ be the all-ones vector in $\FF_2^{2m+1}$ (note that $\ipb{\onev}{v} = 0$ for all $v \in U_m$). For all $u,v \in \FF_2^{2m}$ let
	\begin{align*}
		M_{uv}
		&:= \ipb{\onev+T(u)}{\onev+T(v)}\\
		&= \ipb{\onev}{\onev}
		 + \ipb{\onev}{T(v)}
		 + \ipb{T(u)}{\onev}
		 + \ipb{T(u)}{T(v)}\\
		&= 1+ \sf(u,v).
	\end{align*}
	Since $\sf(u,v) = 1$ if and only if $u$ and $v$ are not joined by an edge in $\spg{2m}{\FF_2}$, matrix $M$ fits $\spg{2m}{\FF_2}$. Since it is the Gram matrix of a set of $(2m + 1)$-dimensional vectors (i.e. the entry at $i,j$ is the inner product of the vector $i$ and vector $j$ for some ordering of the set of vectors), its rank is at most $2m+1$. Therefore, by Haemers' bound, $C_0\big(\spg{2m}{\FF_2}\big) \leq \log (2m + 1)$.
	
	For the matching lower bound, let $e_i$ for $1 \leq i \leq 2m+1$ be the standard basis of $\FF_2^{2m+1}$, and let $f_i := e_i + \onev$. Then  $\ipb{f_i}{f_j} = 1-\delta_{ij}$ so $f_i \in U_m$ and	$\sf\big(T^{-1}(f_i),T^{-1}(f_j)\big) = 1-\delta_{ij}$. Therefore $\{T^{-1}(f_i): i=1,\ldots,2m+1\}$ is an independent set of size $2m+1$ in $\spg{2m}{\FF_2}$, so $\alpha\big(\spg{2m}{\FF_2}\big) \geq 2m + 1$.

	Hence, $C_0\big(\spg{2m}{\FF_2}\big) = \log \alpha\big(\spg{2m}{\FF_2}\big) = \log (2m + 1)$ and the upper bound on the zero-error capacity is attained by a code of block length one.
\end{proof}

\subsection{Entanglement-assisted capacity} \label{sect:Entanglement}

In this section we will establish the entanglement-assisted capacity of $\spg{6}{\FF_2}$. Our main tool is Theorem~\ref{thm:Protocol} together with some combinatorial results.
\begin{definition}
A $d$-dimensional \defemph{orthonormal representation} of graph $G = (V,E)$ is a function $\phi: V \rightarrow \CC^d$ that assigns unit vectors to the vertices of $G$ such that for each edge $uv \in E$ vectors $\vre{u}$ and $\vre{v}$ are orthogonal.
\end{definition}
\noindent The following theorem appeared in \cite{CLMW2} but for completeness we include the proof here.
\begin{theorem}[\cite{CLMW2}]
If graph $G$ has an orthonormal representation in $\CC^d$ and its vertices can be partitioned into $k$ disjoint cliques each of size $d$ then $C_0^{\mathrm{E}}(G) = \log k$.
\label{thm:Protocol}
\end{theorem}

\begin{proof}
	Figure~\ref{fig:part-prot} describes a protocol that uses a rank-$d$ maximally entangled state to send one of $k$ messages with zero error by a single use of the channel, proving that $C_0^{\mathrm{E}}(G) \geq k$.
	Removing edges from $G$ cannot decrease the Lov\'{a}sz number and in this way we can obtain the graph which is the strong product of the empty graph on $k$ vertices with the clique of size $d$ (i.e. the disjoint union of $k$ $d$-cliques). This graph has Lov\'{a}sz number $k$ so $\vartheta(G) \leq k$. Since $C_0^{\mathrm{E}} (G) \leq \log \vartheta(G)$ \cite{Beigi, DSW}, this provides the matching upper bound.
\end{proof}

\begin{figure*}
	\includegraphics[width=\textwidth]{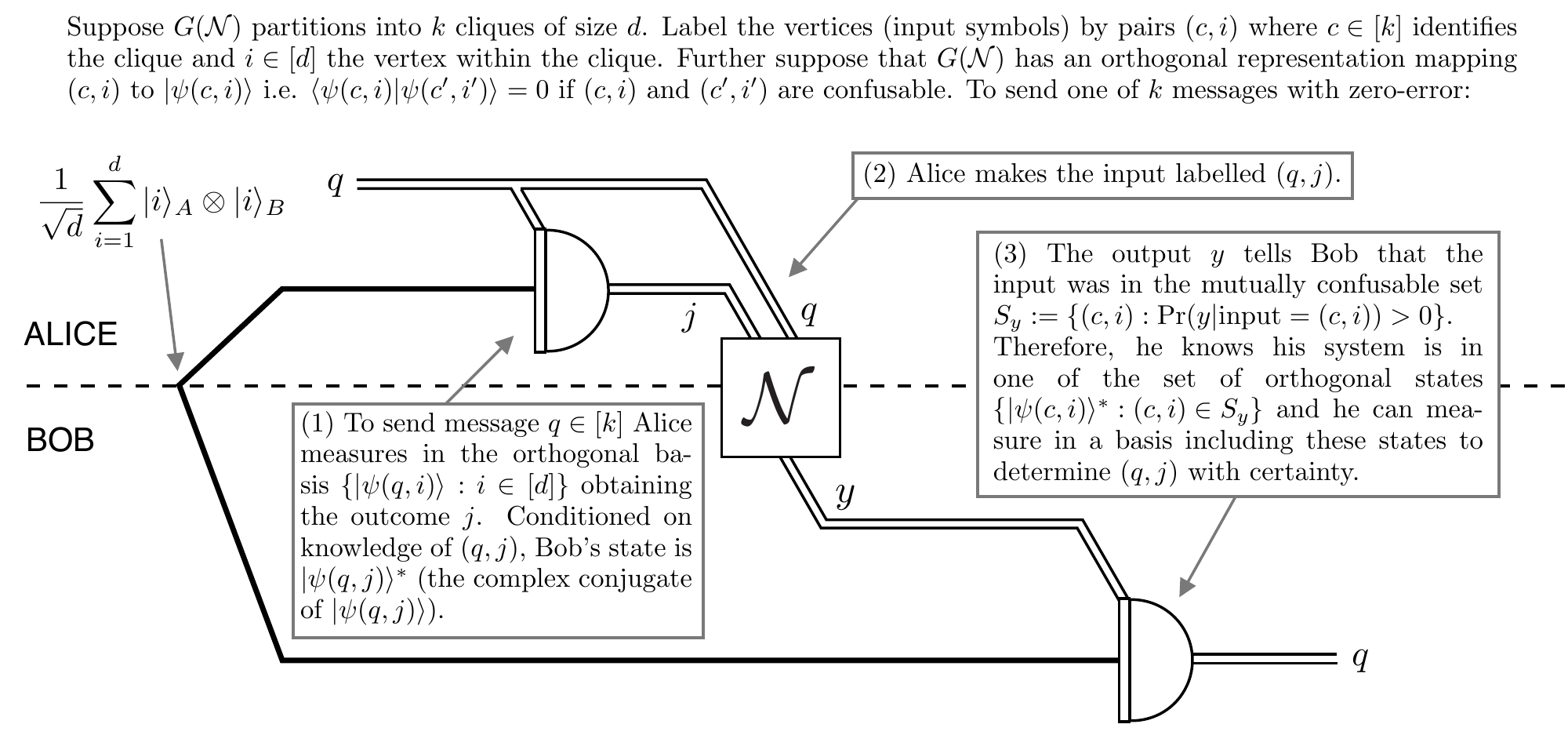}
	\caption{If the confusability graph of $\mc{N}$ can partitioned into $k$ cliques of size $d$ then $C_0(\mc{N}) \leq k$ by the Lov\'{a}sz bound. If it also has a $d$-dimensional orthonormal representation then this rate can be \emph{achieved} by the entanglement-assisted zero-error code (of block length one) described in this figure.}\label{fig:part-prot}
\end{figure*}

\begin{lemma}
The $2^{2m} - 1$ vertices of $\spg{2m}{\FF_2}$ can be partitioned into $2^m + 1$ cliques of size $2^m - 1$.
\label{lem:Cliques}
\end{lemma}

\begin{proof}
	Such a partition of the symplectic graph is known as a \emph{symplectic spread}, and is well-known to exist (see for example~\cite{Dye}). We give a simple construction from~\cite{Spreads} below. Another proof in terms of commuting sets of Pauli operators is given in \cite{MUBs,MUBs2}.

	Let $N = 2^m$, and identify the vertices of $\spg{2m}{\FF_2}$ with the non-zero vectors in $\FF_N^2$. Consider the following symplectic form on $\FF_N^2$:
	\begin{equation}
		\sf_N\big((w,x), (y,z)\big) = \Tr(wz + xy),
	\end{equation}
where $\Tr: \FF_N \rightarrow \FF_2$ is the \emph{finite field trace} defined as $\Tr(a) := a + a^2 + a^{2^2} + \dotsc + a^{2^{m-1}}$. As $\ipb{x}{y}_N := \Tr(xy)$ is a non-degenerate inner product in $\FF_N$, the form $\sf_N$ is also non-degenerate. Hence, the symplectic spaces $(\FF_2^m, \sf)$ and $(\FF_N^2, \sf_N)$ are isomorphic. We will describe the partition for the later space.

	Denoting the multiplicative group (of order $N - 1$) in $\FF_N$ by $\FF_N^{\times} := \FF_N \setminus \set{0}$, the cells of a partition of the non-zero elements of $\FF_N^2$ are:
	\begin{align*}
	\pi_a & = \{(x,ax): x \in \FF_N^{\times}\} \qquad (a \in \FF_N), \\
	\pi_{N+1} & = \{(0,x): x \in \FF_N^{\times}\}.
	\end{align*}
	It is easy to check that these $N+1$ cells of size $N-1$ partition $\FF_N^2$. Moreover, if $(x,ax)$ and $(y,ay)$ are in the same cell, then
	\begin{equation}
		\sf\big((x,ax), (y,ay) \big) = \Tr(xay + axy) = \Tr(0) = 0.
	\end{equation}
	Therefore each cell is a clique.
\end{proof}

\newcommand{\f}{\normalsize}
\newcommand{\ft}{\f\texttt}

\begin{lemma}[\cite{E7,Purbhoo}]
$\spg{6}{\FF_2}$ has an orthonormal representation in $\RR^7$.
\label{lem:E7}
\end{lemma}

\newcommand{\lattice}{\mc{L}} 

The entire representation, grouped into 9 complete (unnormalised) orthogonal bases, is given in a table in Appendix~\ref{apx:Representation} and it suffices to check that this has the desired properties to establish the result. Interestingly, it consists of vectors from the \emph{root system} $E_7$ and it is possible to give a more insightful description and proof of the representation in relation to this. Such a proof is given in Appendix~\ref{apx:E7Proof}.

Since the $63$ vertices of $\spg{6}{\FF_2}$ partition into $9$ cliques each of size $7$ (Lemma~\ref{lem:Cliques}), it follows by Theorem~\ref{thm:Protocol} that $C_0^{\mathrm{E}}\big(\spg{6}{ \FF_2}\big) = \log 9$, whereas we already established that $C_0\big(\spg{6}{ \FF_2}\big) = \log(2\times3 + 1) = \log 7$. This concludes the proof of Theorem \ref{mainthm}.

\subsection{The connection to $E_7$}

A deeper coincidence underlies the orthogonal representation of $\spg{6}{\FF_2}$ by roots of $E_7$. The automorphism group of $\spg{2m}{\FF_2}$ is the \emph{symplectic group} $Sp(2m, \FF_2)$, which is the group of linear maps on $\FF_2^{2m}$ which preserve the symplectic form. This group is isomorphic to quotient $W(E_7)/\{\pm \1\}$, where $W(E_7)$ is the \emph{Weyl group} of $E_7$.

\section{Relationship to the normal capacity}

Given a classical channel $\mathcal{N}$, its standard classical capacity $C(\mc{N})$ cannot be increased by the use of entanglement or even arbitrary non-signalling correlations between the sender and receiver~\cite{BSST}. The standard capacity and the assisted and unassisted zero-error capacities are related by
\begin{equation}
	C_0(\mc{N}) \leq C_0^{\mathrm{E}}(\mc{N}) \leq C(\mc{N}).
\end{equation}


\begin{theorem}
	Given a graph $G$ which satisfies the premises of Theorem \ref{thm:Protocol} (partitions into $k$ cliques of size $d$ and has an orthonormal representation in dimension $d$) and is also vertex-transitive, one can construct a channel $\mc{N}$ whose normal capacity $C(\mc{N})$ and $C_0^{\mathrm{E}}(\mc{N})$ are both equal to $C_0^{\mathrm{E}}(G)$ and are both achieved by a block-length one entanglement-assisted zero-error code.
\end{theorem}
\begin{proof}
	Let $X$ be the vertices of $G$ and let $Y$ be the set of \emph{all} cliques of size $d$ in $G$. Since $G$ is vertex-transitive, each vertex is contained in the same number $m$ of cliques from $Y$. Counting the number of pairs in the set $\set{(x,y) \in X\times Y: x \in y}$ in two ways we have
	\begin{equation}
		\bigl| \{(x,y) \in X\times Y: x \in y\} \bigr| = |Y|d = |X|m.
	\end{equation}
	Let $\mc{N}$ be the channel which on the input $x \in X$ produces an output uniformly at random from the set $\{y : x \in y\} \subseteq Y$. Using the analysis of Section~16 of \cite{Shannon1},
	\begin{equation}
		C(\mc{N}) = \log \frac{|Y|}{m} = \log \frac{|X|}{d} = \log k.
	\end{equation}
	$G(\mc{N})$ also partitions into $k$ cliques of size $d$ and, since it is a subgraph of $G$, has an orthonormal representation in dimension $d$. Therefore $C_0^{\mathrm{E}}(\mc{N}) = \log k$ and, furthermore, this rate is achieved by the block-length one entanglement-assisted protocol of Figure \ref{fig:part-prot}.
\end{proof}

The symplectic graphs are all vertex transitive so, remarkably, the channel constructed in this way for $\spg{6}{\FF_2}$ has $C = C_0^{\mathrm{E}} = \log 9$, even though there is a positive lower bound on the error probability for \emph{classical} codes with any rate greater than $\log 7$ (as well as an upper bound, both decaying exponentially with $n$) \cite{BGS}.

\section{Graphs from $E_8$ and other root systems}

We define the orthogonality graph of a root system $R$ as follows. The vectors of $R$ occur in antipodal pairs $\{v,-v\}$; the vertices $V(R)$ of the graph are the $|R|/2$ rays spanned by these antipodal pairs. Two vertices are adjacent if and only if their rays are orthogonal. The graph of Section~\ref{sect:Main} is precisely the orthogonality graph of $E_7$. This raises the question of whether a channel whose confusability graph is the orthogonality graph of another irreducible root system can also exhibit a gap between the classical and entanglement-assisted zero-error capacities. We find that the orthogonality graph of $E_8$ provides a second example of such a gap, and furthermore, the ratio between the assisted and classical capacities is larger in this case.

The irreducible root systems consist of the infinite families $A_n$, $B_n$, $C_n$, $D_n$ where $n \in \mathbb{N}$, and the exceptional cases $E_6$, $E_7$, $E_8$, $F_4$, and $G_2$ (see~\cite{Humphreys}). We show that for all of the infinite families, and for $G_2$, there is no gap between the independence number $\alpha$ and the Lov\'asz number $\vartheta$, so $C_0^{\mathrm{E}} = C_0$ for these graphs. However, the orthogonality graph of $E_8$ provides a second example of a gap between the classical and entanglement-assisted capacity. For $E_8$, we show that $C_0 \leq \log 9$ while $C_0^{\mathrm{E}} =\log 15$. It is interesting to note that the graph used in~\cite{CLMW} is precisely the orthogonality graph of $F_4$ and while we know the entanglement-assisted capacity of this graph, we still do not know its unassisted capacity or whether it is smaller than the assisted one. We do not give either capacity for the graph of $E_6$.

In what follows the name of the root system is also used as the name of the orthogonality graph and $e_i$ denotes the $i$-th standard basis vector. We ignore correct normalization of the root vectors for simplicity, since it clearly doesn't affect the orthogonality graph.

\subsection*{Root system $E_8$}

\begin{equation}
V(E_8) = \{e_i \pm e_j: 1 \leq i < j \leq 8\} \cup \big\{(x_1,\ldots,x_8) : x_i = \pm 1, \textstyle\prod_{i=1}^8 x_i = 1\big\}
\end{equation}
As pointed out in~\cite{Panigrahi}, $E_8$ is the graph whose vertices are the non-isotropic points in the ambient projective space of the polar space $O^+(8, \FF_2)$, with vertices adjacent if they are orthogonal with respect to the associated bilinear form. The ambient projective space of $O^+(2m,\FF_2)$ is $\mathbb{P} \FF_2^{2m}$. Since the bilinear form associated with $O^+(2m,\FF_2)$ is symplectic, it follows immediately that $E_8$ is an induced subgraph of the symplectic graph $\spg{2m}{\FF_2}$ with $m = 4$. By Theorem~\ref{sc},
\begin{equation}
C_0(E_8) \leq C_0\big(\spg{8}{\FF_2}\big) = \log 9.
\end{equation}

On the other hand, let $N = 16$ and identify the vertices of $\spg{8}{\FF_2}$ with the non-zero vectors of $\FF_N^2$. Then we may choose the quadratic form of $O^+(8,\FF_2)$ to be $(x,y) \mapsto \Tr(xy)$, where $\Tr: \FF_N \rightarrow \FF_2$ is the finite field trace. The polarization of this quadratic form is the symplectic form $\sf\big((w,x),(y,z) \big) = \Tr(wz+xy)$. With this choice, the vertices of $E_8$, i.e., the non-isotropic vectors in $\spg{8}{\FF_2}$, are those $(x,y) \in \FF_N^2$ such that $\Tr(xy) = 1$. Now, consider the partition of vertices into cliques given in Lemma~\ref{lem:Cliques}, restricted to the vertices of $E_8$:
\begin{equation}
	\pi_a = \{(x,ax): x \in \FF_N^{\times}, \Tr(ax^2) = 1\}, \qquad (a \in \FF_N^\times). \\
\end{equation}
Recall that $\Tr(ax^2) = \Tr(a^2x^4) = \ldots = \Tr(a^8x)$. For each $a \in \FF_N^\times$, there are exactly $8$ choices of $x \in \FF_N^\times$ such that $\Tr(a^8x) = 1$. Therefore, $\{\pi_a: a \in \FF_N^\times\}$ is a partition of the vertices of $E_8$ into 15 cliques of size $8$. By Theorem~\ref{thm:Protocol},
\begin{equation}
C_0^{\mathrm{E}}(E_8) = \log 15.
\end{equation}


\subsection*{Root system $A_n$ ($n \geq 1$)}

\begin{equation}
V(A_n) = \{e_i-e_j: 1 \leq i < j \leq n+1\}.
\end{equation}
This graph is isomorphic to the Kneser graph $\KG{n+1}{2}$. By a result of Lov\'asz (Theorem 13 of~\cite{Lovasz}),
\begin{equation}
\al(A_n) = \vartheta(A_n) = \Theta(A_n) = n.
\end{equation}
\subsection*{Root system $D_n$ ($n \geq 4$)}
\begin{equation}
V(D_n) = \{e_i \pm e_j: 1 \leq i < j \leq n\}.
\end{equation}
Note that the vertices $\{e_i-e_j: 1 \leq i < j \leq n\}$ induce a subgraph isomorphic to $A_{n-1} \cong \KG{n}{2}$. Also note that $e_i + e_j$ and $e_i - e_j$ are adjacent and have the same neighbourhood (apart from themselves). It follows that $D_n$ is isomorphic to $\KG{n}{2} \boxtimes K_2$, the strong graph product of a Kneser graph and a complete graph on $2$ vertices. By Theorem 7 of~\cite{Lovasz},
\begin{equation}
\Theta(D_n) = \Theta(\KG{n}{2}) \Theta(K_2) = n-1.
\end{equation}
Since $\{e_1 - e_j: 2 \leq j \leq n\}$ is an independent set of size $n-1$, it follows that 
\begin{equation}
\al(D_n) = \vartheta(D_n) = \Theta(D_n) = n-1.
\end{equation}

\subsection*{Root system $B_n$ ($n \geq 2$)}

\begin{equation}
V(B_n) = \{e_i \pm e_j: 1 \leq i < j \leq n\} \cup \{e_i : 1 \leq i \leq n\}.
\end{equation}
To find the Lov\'asz number we consider even and odd $n$ separately. When $n$ is odd, partition the vertices into sets $\pi_1,\ldots,\pi_n$, where
\begin{equation}
\pi_k = \{e_i \pm e_j: i < j, i+j \equiv 2k \pmod{n}\} \cup \{e_k\}. 
\end{equation}
Each $\pi_k$ is a clique of size $n$. When $n$ is even, partition the vertices into sets $\pi_1,\ldots,\pi_n$, where
\begin{align*}
\pi_k & = \set{e_i \pm e_j: i < j, i+j \equiv 2k \pmod{n-1}} \cup \set{e_k \pm e_n} \quad (k \leq n-1); \\
\pi_n & = \set{e_1,\ldots,e_n}.
\end{align*}
Again each $\pi_k$ is a clique of size $n$. In either case, we have partitioned the graph into $n$ cliques of size $n$. By Theorem~\ref{thm:Protocol}, $\Theta(B_n) = n$.

Since $\{e_1 - e_j: 2 \leq j \leq n\} \cup \{e_1\}$ is an independent set of size $n$, we have
\begin{equation}
\al(B_n) = \vartheta(B_n) = \Theta(B_n) = n.
\end{equation}

\subsection*{Root system $C_n$ ($n \geq 2$)}

\begin{equation}
V(C_n) = \{e_i \pm e_j: 1 \leq i < j \leq n\} \cup \{2e_i : 1 \leq i \leq n\}.
\end{equation}
$C_n$ has the same orthogonality graph as $B_n$.

\subsection*{Root system $G_2$}

\begin{equation}
V(G_2) = \{(1,-1,0),(1,0,-1),(0,1,-1),(1,1,-2),(1,-2,1),(-2,1,1)\}.
\end{equation}
By inspection $G_2$ has an independent set of size $3$ and can be partitioned into $3$ cliques of size $2$. By Theorem~\ref{thm:Protocol},
\begin{equation}
\al(G_2) = \vartheta(G_2) = \Theta(G_2) = 3.
\end{equation}

\section{Conclusion}

We have shown that it is possible for entanglement to increase the asymptotic rate of zero-error classical communication over some classical channels. This is quite different from the situation for families of codes which only achieve arbitrarily small error rates asymptotically. The best rate that can be achieved by classical codes in this context is the Shannon capacity and entanglement \emph{cannot} increase this rate. The entanglement-assisted capacity has a simple formula which reduces to the formula for the Shannon capacity when the channel is classical \cite{CRST}.

It is interesting to note that in every example of a graph with $M_0(G) > M_0^{\mathrm{E}}(G)$ found to date, the entanglement-assisted capacity is attained by a code of block length one. This certainly is not true of the entanglement-assisted capacities of all graphs. In \cite{ZEIT}, an interesting observation of Arikan is reported: The graph consisting of a five cycle and one more isolated vertex has $\Theta = \sqrt{5}+1$. Since no positive integer power of this quantity is an integer, the capacity is not attained by any finite length block code for this graph. Since the Lov\'{a}sz number of this graph is also $\sqrt{5}+1$ (the Lov\'{a}sz number is additive for disjoint unions of graphs, and was calculated for cycles in \cite{Lovasz}), the same observation is true for the entanglement-assisted capacity, which in this case is equal to the unassisted capacity.

Our result has an interesting interpretation in terms of Kochen-Specker (KS) proofs of non-contextuality. Such a proof specifies a set of complete, projective measurements, with some projectors in common, such that there is no way to consistently assign a truth value to each projector. An assignment is consistent if (a) precisely one projector in each measurement is ``true'' and (b) no two ``true'' projectors are orthogonal.

Ruuge \cite{ruuge} shows that the root systems $E_7$ and $E_8$ can be used to construct KS proofs using computer search to nullify the possibility of a consistent assignment. This is a corollary of our results, but our proof is analytic due to the novel application of the Haemers bound. In fact, the use of the Haemers bound provides a whole sequence of KS proofs which are increasingly strong in the following quantitative sense: For the set of $9^n$ measurements which are obtained by tensoring together $n$ of Alice's 9 measurements, only $7^n$ can be assigned values in accordance with property (a) before property (b) must be violated.

Three main avenues for further research are apparent to us. First, is it possible to give a general algorithm to compute $C_0^{\mathrm{E}}$? More specific related problems include determining whether $C_0^{\mathrm{E}}(G)/C_0(G)$ can be arbitrarily large, and whether there are graphs where $C_0^{\mathrm{E}}(G)$ is strictly less than $\log \vartheta (G) $.

Secondly, we have already shown that there are some connections to multi-prover games and to non-contextuality, but we feel that a deeper understanding of these connections is possible and desirable. For example, the application of our result to KS proofs mentioned above suggests some stronger notion of non-contextuality in quantum mechanics. 

Finally, our work on entanglement-assisted zero-error codes can be placed in the wider context of using entanglement to reduce decoding error in finite block length coding of classical information for classical channels (demonstrating this effect is even experimentally feasible \cite{exp}), and characterising this phenomenon presents an even wider set of questions.
\\
\\
\noindent{\bf Acknowledgements}\\
We would like to thank Andrew Childs, Richard Cleve, David Roberson, Simone Severini, and Andreas Winter for useful discussions.
This work was supported by NSERC, QuantumWorks, CIFAR, CFI, and ORF. Aidan Roy acknowledges support by a UW/Fields Institute Award.

\bibliographystyle{unsrturl} 

\bibliography{E7}

\newpage

\appendix

\section{The orthogonal representation of $\spg{6}{\FF_2}$ in full} \label{apx:Representation}

\begin{table}[h!]
	\center
	\includegraphics[scale=0.95]{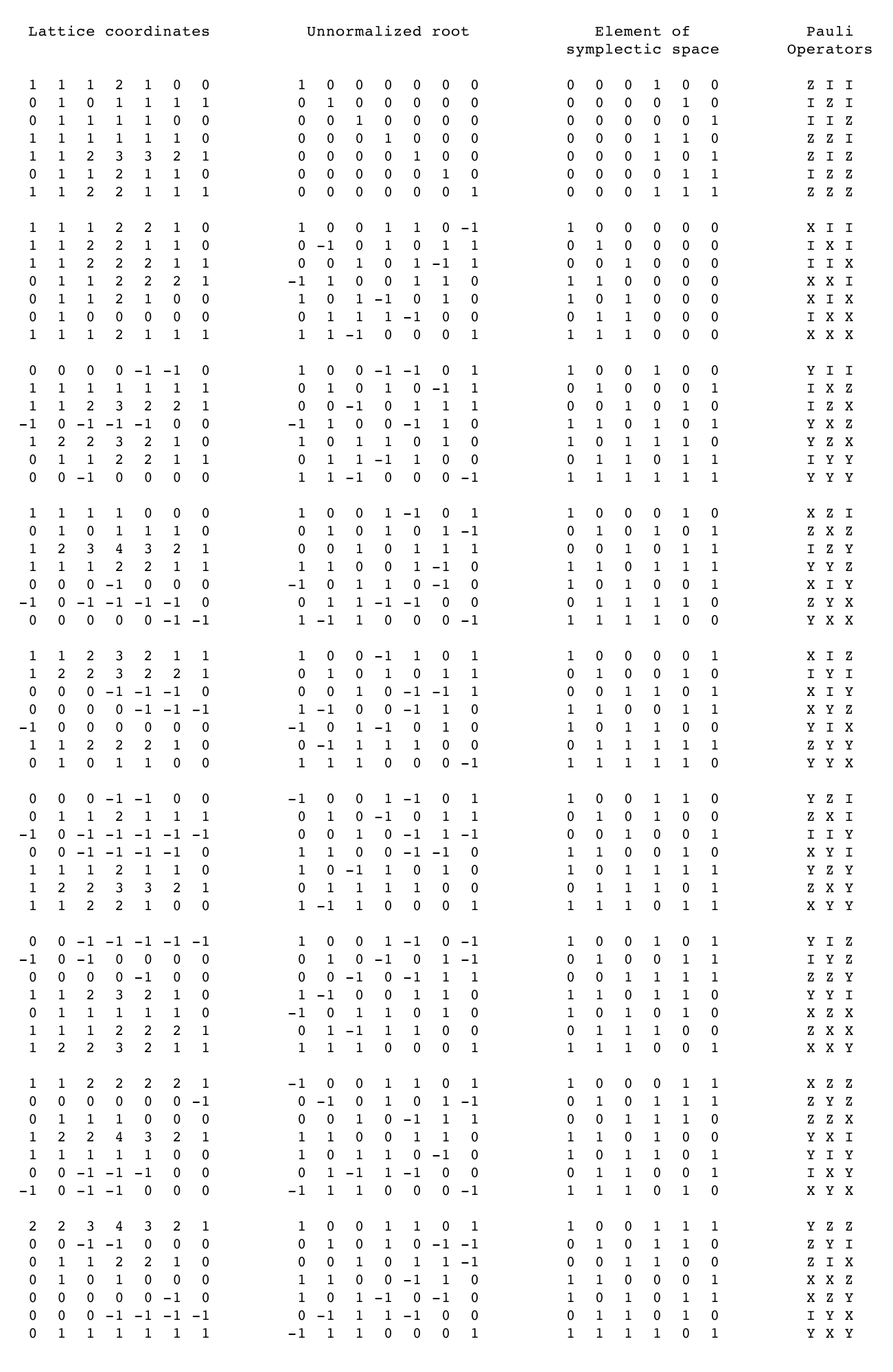}
	\caption{Full listing of the orthogonal representation of $\spg{6}{\FF_2}$ grouped into $9$ complete orthogonal basis. For each row, the first column shows the lattice coordinate $(n_1, \dotsc, n_7)$, followed by the real coordinates of the corresponding root $\sum_{i=1}^7 n_i \alpha_i$, followed by the corresponding element $\sum_{i=1}^7 n_i v_i$ of $\mathbb{F}^6_2$.  The last column rephrases the $\mathbb{F}^6_2$ element as a $3$-qubit Pauli operator.}
\end{table}

\section{The proof of Lemma~\ref{lem:E7}} \label{apx:E7Proof}

\begin{figure}[h]
  \centering
  \includegraphics{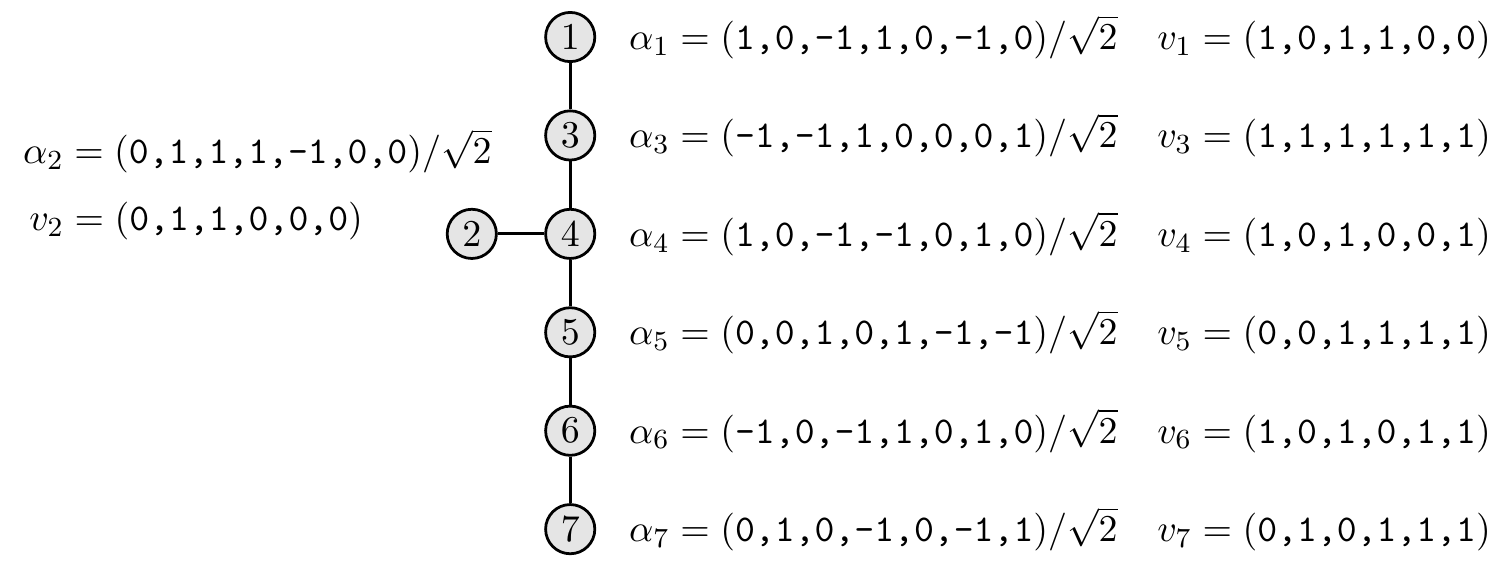}
  \caption{The Dynkin diagram of $E_7$. Each node is a simple root $\alpha_i \in \RR^7$ and the edges determine their inner products according to equation~(\ref{eq:inner products}). The corresponding vectors $v_i \in \FF_2^6$ used in the definition of $\kappa$ (\ref{kappa}) are also shown.}
  \label{fig:Dynkin}
\end{figure}

 We first give some basic definitions and facts about $E_7$ and its lattice. Let $\alpha_1, \dotsc, \alpha_7 \in \RR^7$ be the vectors given in Figure~\ref{fig:Dynkin}, known as \emph{simple roots} of $E_7$. Their inner products are encoded by \emph{Dynkin diagram} shown in Figure~\ref{fig:Dynkin} as follows:
\begin{equation}
  \ipb{\alpha_i}{\alpha_j} =
  \begin{cases}
    \phantom{-}0 & \text{if nodes $i \neq j$ are not connected}, \\
              -1 & \text{if nodes $i \neq j$ are connected}, \\
    \phantom{-}2 & \text{if $i=j$}.
  \end{cases}
  \label{eq:inner products}
\end{equation}
All integer linear combinations of the simple roots form the {\bf $\boldsymbol{E}_7$ lattice}
\begin{equation}
  \lattice := \bigl\{ \textstyle \sum_{i=1}^7 n_i \alpha_i : n_1, \dotsc, n_7 \in \ZZ \bigr\}.
\end{equation}
For $\gamma = \sum_{i=1}^7 n_i \alpha_i$, let $\lc{\gamma} := (n_1,n_2,\ldots,n_7)$ denote the lattice coordinates of $\gamma$. The inner product between two lattice vectors is
\[
	\ipb{\gamma}{\delta} = \sum_{i,j}\lc{\gamma}_i\lc{\delta}_j \ipb{\alpha_i}{\alpha_j}
	= 2 \sum_i \lc{\gamma}_i\lc{\delta}_i - 2 \sum_{\{i,j\} \in \de} \lc{\gamma}_i\lc{\delta}_j
\]
where $\de = \{\{1,3\},\{2,4\},\{3,4\},\{4,5\},\{5,6\},\{6,7\}\}$ is the set of edges in the Dynkin diagram. Note that the inner product $\ipb{\gamma}{\delta}$ is an \emph{even integer} for all $\gamma$, $\delta$.

The \emph{root system $E_7$} is the set vectors of norm $\sqrt{2}$ in $\lattice$:\footnote{One can check that this agrees with the more common definition of $E_7$ as the orbit of $\alpha_1$ under the reflection group $\langle R_1, \dotsc, R_7 \rangle$, where $R_i := \1 - 2 \hat{\alpha}_i \hat{\alpha}_i \tp$ and $\hat{\alpha}_i$ is the unit vector in direction $\alpha_i$.}
\begin{equation}
  E_7 := \set{\gamma \in \lattice : \ipb{\gamma}{\gamma} = 2}.
\end{equation}
In terms of lattice coordinates, the condition $\ipb{\gamma}{\gamma} = 2$ can be expressed as
\begin{equation}
	  r(\lc{\gamma}) := \lc{\gamma}_1^2 + \lc{\gamma}_2^2 + \lc{\gamma}_3^2 + \lc{\gamma}_4^2 + \lc{\gamma}_5^2 + \lc{\gamma}_6^2 + \lc{\gamma}_7^2
	  - \lc{\gamma}_1 \lc{\gamma}_3
	  - \lc{\gamma}_2 \lc{\gamma}_4
	  - \lc{\gamma}_3 \lc{\gamma}_4
	  - \lc{\gamma}_4 \lc{\gamma}_5
	  - \lc{\gamma}_5 \lc{\gamma}_6
	  - \lc{\gamma}_6 \lc{\gamma}_7 = 1.
	  \label{eq:poly}
\end{equation}

Following \cite[Section 4.6]{E7} and \cite[Section 3.2]{Purbhoo} we we define $\kappa: \mc{L} \to \FF_2^6$ by
\begin{equation}
	\kappa(\gamma) := \sum_{i = 1}^{7} \lc{\gamma}_i v_i \mod 2.
	\label{kappa}
\end{equation}
The $v_i$ (defined in the figure above) are chosen so that $\sigma(\kappa(\alpha_i), \kappa(\alpha_j)) = \ipb{\alpha_i}{\alpha_j} \mod 2$. This extends to all lattice vectors by linearity of $\kappa$:
\begin{equation}\label{compat}
	\sigma(\kappa(\gamma), \kappa(\delta)) = \sum_{i,j} \lc{\gamma}_i\lc{\delta}_j \sigma(\kappa(\alpha_i), \kappa(\alpha_j)) = \sum_{i,j} \lc{\gamma}_i\lc{\delta}_j \ipb{\alpha_i}{\alpha_j} = \ipb{\gamma}{\delta} \mod 2.
\end{equation}

We can write $\kappa(\gamma) = \kappa'(\gamma) \mod 2$ where $\kappa': \mc{L} \to \mathbb{Z}^6$ is defined by $\kappa'(\gamma) := \sum_{i = 1}^{7} \lc{\gamma}_i v_i$ with the $v_i$ are treated as vectors in $\mathbb{Z}^6$. It is easily checked that the kernel of $\kappa'$ is the set $\{ m \w : m \in \mathbb{Z} \}$ where $\w = \alpha_2 + \alpha_5 + \alpha_7$ (i.e. $\lc{\w} = (0,1,0,0,1,0,1)$). Therefore,
\begin{lemma}
	$\kappa(\gamma) = 0$ iff $\gamma = 2 \delta + t \w$ for some $\delta \in \mc{L}$ and $t \in \{0,1\}$.
\end{lemma}

\begin{lemma}
	If $\kappa(\gamma) = 0$ then $\gamma$ is not a root.
\end{lemma}
\begin{proof}
	If $\kappa(\gamma) = 0$ then $\gamma = 2\delta + t \w$ for some $\delta \in \mc{L}$, $t \in \{0,1\}$, so $\ipb{\gamma}{\gamma} = 4\ipb{\delta}{\delta} + 4\ipb{\delta}{\w} + \ipb{\w}{\w}$.
	The first two inner products are even integers and $\ipb{\w}{\w} = 6$, so for some integer $m$, $\ipb{\gamma}{\gamma} = 8m + 6 \neq 2$ and can't be a root by definition.
\end{proof}

\begin{lemma}
$\alpha, \beta \in E_7$ and $\kappa(\alpha) = \kappa(\beta)$ iff $\beta = \pm\alpha$.
\end{lemma}
\begin{proof}
$\kappa(\alpha) = \kappa(\beta)$ iff $\beta - \alpha = t \w + 2 \delta$ for some $\delta \in \mc{L}$, $t \in \{0,1\}$. We can rule out the case where $t = 1$ because, if it were
\[
	r(\lc{\beta}) = r(\lc{\alpha}) + 2r(\lc{\delta}) + r(\lc{w}) = r(\lc{\alpha}) + 1 \mod 2
\]
since $r(\lc{\w}) = 3$. Then, $\alpha$ and $\beta$ cannot both be roots due to equation (\ref{eq:poly}), which must necessarily hold modulo two. Therefore, $t = 0$ and $\beta = \alpha + 2 \delta$. Since $\alpha$ is a root, the condition that $\beta$ is also a root
\[
 	\ipb{\beta}{\beta} = \ipb{\alpha}{\alpha} + 4\ipb{\alpha}{\delta} + 4\ipb{\delta}{\delta} = 2,
\]
reduces to $\ipb{\delta}{\alpha} = - \ipb{\delta}{\delta}$. By the Cauchy-Schwarz inequality, $|\ipb{\alpha}{\delta}|^2 \leq \ipb{\alpha}{\alpha} \ipb{\delta}{\delta} = 2\ipb{\delta}{\delta}$, or equivalently, $|\ipb{\alpha}{\delta}| \leq 2$, with equality iff  $\delta$ is a scalar multiple of $\alpha$. Since inner products between lattice vectors are even integers, either $\delta = 0$ and $\beta = \alpha$, or $\delta = -\alpha$ and $\beta = -\alpha$.
\end{proof}

There are 126 roots in 63 antiparallel pairs. Let $R$ be a subset of $E_7$ with one root from each pair. We have just shown that both roots in a pair have the same image under $\kappa$, that these images are different for different pairs, and none are equal to 0.  Therefore, the restriction of $\kappa$ to the domain $R$ is a bijection between $R$ and $\mathbb{F}_2^6\setminus\{0\}$ whose inverse determines (by normalising the vectors in $R$) an orthonormal representation of $\spg{6}{\FF_2}$ thanks to the relationship (\ref{compat}).

\end{document}